\documentclass[11pt,conference,compsoc,onecolumn,romanappendices]{IEEEtran}
\usepackage[cmex10]{amsmath}
\usepackage{amssymb,enumerate,bbm}
\usepackage{amsthm,graphicx,float,mathtools}
\usepackage[ruled,boxed]{algorithm2e}
\usepackage{cite}

\usepackage{fancyhdr}
\fancypagestyle{plain}{%
  \fancyhf{}
  \fancyfoot[C]{\iffloatpage{}{\thepage}}
  }
\newcommand{\nin}{\not\in}
\newcommand{\nocomma}{}
\newcommand{\tmem}[1]{{\em #1\/}}
\newcommand{\tmop}[1]{\ensuremath{\operatorname{#1}}}
\newcommand{\tmstrong}[1]{\textbf{#1}}

\newenvironment{enumeratenumeric}{\begin{enumerate}[1.] }{\end{enumerate}}
\newenvironment{enumerateroman}{\begin{enumerate}[i.] }{\end{enumerate}}
\newenvironment{itemizedot}{\begin{itemize} }{\end{itemize}}
\newenvironment{tmindent}{\begin{tmparmod}{1.0em}{1.0em}{0pt} }{\end{tmparmod}}
\newenvironment{fkindent}{\begin{tmparmod}{0.0em}{1.0em}{0pt} }{\end{tmparmod}}
\newenvironment{tmparmod}[3]{\begin{list}{}{\setlength{\topsep}{0pt}\setlength{\leftmargin}{#1}\setlength{\rightmargin}{#2}\setlength{\parindent}{#3}\setlength{\listparindent}{\parindent}\setlength{\itemindent}{\parindent}\setlength{\parsep}{\parskip}} \item[]}{\end{list}}
\newtheorem{lemma}{Lemma}
\newtheorem{theorem}{Theorem}

\IEEEoverridecommandlockouts

\title{\LARGE Beyond One Third Byzantine Failures}

\author{
\IEEEauthorblockN{
Wang Cheng\IEEEauthorrefmark{1},\quad
Carole Delporte-Gallet\IEEEauthorrefmark{2},\quad
Hugues Fauconnier\IEEEauthorrefmark{3}\\[0.7em]
Rachid Guerraoui\IEEEauthorrefmark{4},\quad
Anne-Marie Kermarrec\IEEEauthorrefmark{5}}
\thanks{\IEEEauthorrefmark{1}\'{E}cole Polytechnique F\'{e}d\'{e}rale de Lausanne,Switzerland,
\; Email: cheng.wang@epfl.ch}
\thanks{\IEEEauthorrefmark{2}LIAFA-Universit\'{e} Paris-Diderot, Paris, France,
\; Email: cd@liafa.univ-paris-diderot.fr}
\thanks{\IEEEauthorrefmark{3}LIAFA-Universit\'{e} Paris-Diderot, Paris, France,
\; Email: hf@liafa.univ-paris-diderot.fr}
\thanks{\IEEEauthorrefmark{4}\'{E}cole Polytechnique F\'{e}d\'{e}rale de Lausanne, Switzerland,
\; Email: rachid.guerraoui@epfl.ch}
\thanks{\IEEEauthorrefmark{5}INRIA Rennes Bretagne-Atlantique, France,
\; Email: anne-marie.kermarrec@inria.fr}
}

\date{}

\begin{document}

\pagenumbering{gobble}

\maketitle

\begin{abstract}

The Byzantine agreement problem requires a set of $n$ processes to agree on
a value sent by a transmitter, despite a subset of $b$ processes behaving in
an arbitrary, i.e. Byzantine, manner and sending corrupted messages to all
processes in the system. It is well known that the problem has a solution in a
(an eventually) synchronous message passing distributed system iff the number
of processes in the Byzantine subset is less than one third of the total
number of processes, i.e. iff $n > 3b+1$. The rest of the processes are
expected to be correct: they should never deviate from the algorithm assigned
to them and send corrupted messages. But what if they still do?

We show in this paper that it is possible to solve Byzantine agreement even
if, beyond the $ b$ ($< n/3 $) Byzantine processes, some of the other  processes
also send corrupted messages, as long as they do not send them to all.  More
specifically, we generalize the classical Byzantine model and consider that
Byzantine failures might be partial. In each communication step, some of the
processes might send corrupted messages to a subset of the processes. This
subset of processes - to which corrupted messages might be sent - could change
over time. We compute the exact number of processes that can commit such
faults, besides those that commit classical Byzantine failures, while still
solving Byzantine agreement. We present a corresponding Byzantine agreement
algorithm and prove its optimality by giving resilience and complexity
bounds.

\end{abstract}

\vspace{2cm}
\begin{center}
\Large This paper is a regular submission.\\
\Large The paper is a student paper.
\end{center}

\newpage
\pagenumbering{arabic}\setcounter{page}{1}

{\newpage}

\section{Introduction}

Pease, Shostak and Lamport introduced the Byzantine model in
their landmark papers~{\cite{lamport1982byzantine,pease1980reaching}}.
A Byzantine process is defined as a process that can arbitrarily
deviate from the algorithm assigned to it and send corrupted messages to other 
processes.  They considered a synchronous model and proved that agreement is
achievable with a fully connected network if and only if the number of
Byzantine processes is less than one third of the total number of processes.
Dolev extended this result to general networks, in which the connectivity
number is more than twice the number of faulty processes
{\cite{dolev1982byzantine}}. The early work on Byzantine agreement is well
summarized in the survey by Fischer {\cite{fischer1983consensus}}.

Several approaches have been proposed to circumvent the impossibility of reaching
Byzantine agreement in an asynchronous context
{\cite{fischer1985impossibility}}.
The eventually synchronous model was presented in {\cite{dwork1988consensus}}:
an intermediate model between synchronous and asynchronous models, allowing some
limited periods of asynchrony. Eventual synchrony is considered weak enough to
model real systems and strong enough to make Byzantine agreement solvable.
Alternative approaches rely on randomized algorithms
{\cite{braud2013fast,dwork1988fault,king2011load,rabin1983randomized}}. As
Karlin and Yao showed in \cite{karlin1986probabilistic}, the \tmem{one third} bound is still
a tight bound for randomized Byzantine agreement algorithms. 

We show in this paper that it is possible to solve Byzantine agreement deterministically even
if, beyond the $ b$ ($< n/3 $) Byzantine processes,  some of the other  processes
also send corrupted messages, as long as they do not send them to all.  We
show that this is possible deterministically, and even in an eventually
synchronous model.  We compute the exact number of processes that can commit
such partial Byzantine faults, besides those that commit classical Byzantine
failures, while still solving Byzantine agreement.  For pedagogical purposes,
we mainly focus in the main paper on the synchronous context and non-signed messages
\cite{lamport1982byzantine,lincoln1993formally}. We discuss signed messages and the
eventually synchronous context in Section \ref{sec-resilience} and the Appendices.

We generalize the classical Byzantine model and
consider that Byzantine failures might be partial. This generalization is, we
believe, interesting in its own right.  In each communication step, some of
the processes might send corrupted messages to a subset of the processes.
The classical Byzantine failure model corresponds to the extreme case  where
this subset is the entire system.  So we consider a system of $n$ processes,
of which  $m$ can be partially faulty.  The processes communicate with each
other directly through a complete network. We assume that each partially
faulty process $p$ is associated with up
to $d$ ($< n - 1$) Byzantine communication links. Such a process $p$ is said to be
{\tmem{$d$-faulty}}. The $d$ Byzantine links are {\tmem{dynamic}}: they may
be different in different communication rounds.
A d-faulty process
somehow means that the local computation of the processes remains correct:
only the communication links related to the faulty processes are
controlled by the adversary - during specific rounds.  This captures practical
situations where processes experience possibly temporary bugs in specific
parts of their code or communication links.
From the component failure model's
view, our generalization is orthogonal to those of
{\cite{tseng2013iterative,santoro2007agreement,amitanand2003distributed}}.

We establish tight bounds on Byzantine agreement in terms of (a) the number of
processes to which corrupted messages can be sent and (b) time complexity,
i.e. the number of rounds needed to reach agreement. Besides basic distributed
computing tools like full information protocols and scenario arguments, we
also introduce and make  use of a new technique we call ``\tmem{View-Transform}'' which
basically enables processes to locally correct partial Byzantine failures
and transform  a classical Byzantine agreement algorithm  into one
that tolerates more than $1/3$ failures. Interestingly, this transformation only
requires  adding a couple more rounds to a classical Byzantine agreement
algorithm, i.e., its time complexity does not grow with the  number of partial
Byzantine faults tolerated.
In fact, by tolerating more than $1/3$ Byzantine failures, our algorithm can
be faster than classical algorithms in the following sense.
In situation where $1/3$ processes are Byzantine, a deterministic
Byzantine algorithm \cite{lamport1982byzantine} need to wait for all correct
processes to communicate, even if some of the communication links between processes
have very large delays.
In our case, these highly delayed links will be viewed as partial failures, and can be totally tolerated.

For a system with $b$ Byzantine processes and $m$ ``d-faulty'' processes, 
Byzantine agreement can be solved among $n$ processes iff
$n > \max \{2 m  + d, 2 d + m, b\}+ 2 b$.    There is thus a clear trade-off between
the number $b$ of Byzantine failures we can tolerate, the number $m$ of  partial
Byzantine failures and $d$.   For instance, the system
could tolerate $1 / 6$ fraction of ``$1$-faulty'' processes in addition to $(1 / 3 - \epsilon)$
Byzantine processes. Tolerating fewer classical Byzantine  failures
would enable us to tolerate many more partial Byzantine ones.
For example, if $b = 0$, we can tolerate up to $n / 2$ ``$1$-faulty'' processes.

The rest of the paper is organized as follows. 
Section \ref{sec-model} describes our partial Byzantine failure model and recalls the Byzantine agreement problem. 
Section \ref{sec-algorithm} presents our Byzantine agreement algorithm in the synchronous context. 
Section \ref{sec-resilience} proves the resilience optimality of our algorithm and also discusses the case where messages are signed.
Section \ref{sec-time} discusses the time optimality of the algorithm.
We conclude by reviewing related work in Section \ref{sec-rel}.
For space limitations, we defer the discussion on early decision and eventual synchrony, 
as well as some correctness proofs to the optional appendices.

\medskip

\section{Model and Definitions}\label{sec-model}

\subsection{Synchronous computations}

We first consider a synchronous message passing distributed system $P$ of $n$
processes. Each process is identified by a unique id $p \in \{0, 1, \ldots, n
- 1\}$.
As in {\cite{lamport1982byzantine,toueg1984simple}}, a
synchronous computation proceeds in a sequence of rounds.\footnote{We
consider eventually synchronous computations in Appendix \ref{sec-esyn}.}
The processes communicate by exchanging messages round by round within a fully
connected point-to-point network. In each round, each process $p$ first sends at
most one message to every other process, possibly to all processes, and then
$p$ receives the messages sent by other processes. The communication channels
are authenticated, i.e. the sender is known to the recipient. Following
{\cite{lamport1982byzantine}}, we consider oral messages\footnote{We discuss the impact of signed messages in Section \ref{sec-resilience}.}
with the following properties:
(a) every message sent is delivered;
(b) the absence of a message can be detected.
In the system, there is a designated process called \tmem{transmitter} which has
an initial input value from some domain $\mathcal{V}$ to transmit to all processes.

We model an algorithm as a set of deterministic automata, one for each
process in the system. Thus, the actions of a process are entirely determined
by the algorithm, the initial value of the transmitter and the messages it receives from others.

\subsection{Failure model}

In short, a d-faulty process $p$ may lie to other processes: in each round, $p$ can send to
a subset of $d$ processes Byzantine messages, i.e., messages that differ from
those that $p$ has to send following its algorithm. We
assume  that up to $m$ ($\geqslant 0$) of the processes are partial
controlled by the adversary (these processes can send  Byzantine messages to
$d$ ($< n - 1$) processes) and  up to $b$ ($\geqslant 0$) are fully controlled
by the adversary. By convention, if $m=0$, we
assume $d=0$ to make our condition simpler to state.

In each round, the adversary chooses up to $d$ communication links from each partial controlled
process that could carry Byzantine messages, while the fully controlled
processes could send Byzantine messages. We call an instance of our system of $n$
processes with $m$ d-faulty processes and $b$ Byzantine processes as a {\tmem{(n, m, d, b)-system}}.
We refer to the correct processes as well as the d-faulty ones as \tmem{non-Byzantine} processes in this paper.

\subsection{Full information algorithms}

We consider full information algorithms
in the sense of {\cite{fischer1982lower,lamport1982byzantine,lynch1996distributed}}.
Every process transmits to all processes its
entire state in each round, including everything it knows about all values
sent by other processes in the previous round. We introduce in the following
a collection of notations (a slight extension of \cite{fischer1982lower}) to establish and prove our results.

We use $P^{l : k}$ to denote the set of strings of process identifiers in
$P$ of length at least $l$ and at most $k$, and $P^k$ to denote the set of
strings of length $k$. An empty string has length $0$. We
use $P^+$ to denote non-empty strings of symbols in $P$ and $P^{\ast}$ to
denote all strings including the empty one. 
We always refer to
$p_0$ as the transmitter in the Byzantine agreement problem,
and $\mathcal{V}$ as the domain of values which processes 
wish to agree on. For convenience, we assume that $\{\bot, 0, 1\} \in \mathcal{V}$
where $\bot$ refers to the empty value.

A {\tmem{$k$-round scenario}} $\sigma$ (in a (n, m, d, b)-system $P$)
describes an execution of the algorithm. Intuitively $\sigma$ describes a
communication scheme admissible for the (n, m, d, b)-system. The scenario 
is determined by the
initial value of each process and the communication scheme.
Given scenario $\sigma$,
$\sigma (p_0 p_1 \ldots p_k)$ represents the value $p_{k - 1}$
tells $p_k$ that $p_{k - 2}$ tells $p_{k - 1}$ ... that $p_0$ tells $p_1$ is
$p_0$'s initial value.
Formally, a {\tmem{$k$-round scenario}} $\sigma$ is a mapping $\sigma : p_0
P^{0 : k} \rightarrow \mathcal{V}$, such that:
\begin{itemize}
  \item $\sigma (p_0)$ is the initial value of transmitter $p_0$.
  \smallskip
  
  \item There are sets $B (\sigma)$ and $D (\sigma)$ of processes (denoting
  the set of Byzantine and d-faulty processes, respectively) such
  that:
  \begin{itemize}\itemsep0em
    \item $|B (\sigma) | \leqslant b$ and $| D (\sigma) | \leqslant m$,
    
    \item for every process $p \nin (B (\sigma) \cup D (\sigma))$: $\sigma
    (wpq) = \sigma (wp)$ for all $q \in P$ and $w \in p_0 P^{0 : k - 2}$,
  
    \item 
     for every process $p \in D(\sigma)$ and round $l$ ($\leqslant k$), there is a set $T$ of at most $d$ processes such that for
    every $q \in P \setminus T$ and every $w \in p_0 P^{l - 2}$ we have
    $\sigma (wpq) = \sigma (wp)$.
  \end{itemize}
\end{itemize}
Note that $\sigma (wpq) \neq \sigma (wp)$ for some strings $w$ of length
$l$ and process $q$ means that $q$ receives a Byzantine message from $p$ in round $l +
1$.

Throughout this paper, we use $\sigma$ to represent a $k$-round scenario for
a (n, m, d, b)-system with transmitter $p_0$, d-faulty processes in $D
(\sigma)$ and Byzantine processes in $B (\sigma)$. Let $\sigma_p (s) = \sigma (s
p)$ for every $s \in p_0 P^{0 : k-1}$. $\sigma_p$ is called the \tmem{view} of $p$. Let $\sigma_{q_1 \ldots q_i} (s) =
\sigma (s q_1 \ldots q_i)$ for every $s \in p_0 P^{0 : k-i}$. $\sigma_{q_1 \ldots q_i}$ is $q_i$'s view of $q_{i
- 1}$'s view ... of $q_1$'s view, or in short $q_i$'s view from $q_1 \ldots
q_i$. Let $\sigma^{p_0 \ldots p_i}$ denote the $(k - i)$-round scenario with transmitter
$p_i$ such that $\sigma^{p_0 \ldots p_i} (p_i s) =
\sigma (p_0 \ldots p_i s)$ for every $s \in P^{0 : k-i}$.
Naturally, $\sigma^{p_0 \ldots p_i}_p$ denotes the view of $p$ with respect
to scenario $\sigma^{p_0 \ldots p_i}$, and $\sigma^{p_0 \ldots p_i}_{q_{1} \ldots
q_j}$ denotes the view of $q_j$ from $q_{1} \ldots q_j$ with respect to
scenario $\sigma^{p_0 \ldots p_i}$.

Let $\mathcal{U}^k$ be the set of
mappings from $p_0 P^{k - 1}$ into $\mathcal{V}$. Any $k$-round algorithm $F$
defined in a (n, m, d, b)-system may be defined on the set of all views;
namely as a function $F$: $\mathcal{U}^k \rightarrow \mathcal{V}$.

\subsection{The Byzantine agreement problem}\label{model-BA}

We  address in this paper the problem of {\tmem{Byzantine agreement}}
 (also called the Byzantine generals problem in \cite{lamport1982byzantine}).
Each process has an output register which records the outcome of
the computation. We assume that the initial value of this register is 
$nil \notin \mathcal{V}$ and that this  output register can be written at most once.

Let $F$ be a $k$-round algorithm and the output is a value in $\mathcal{V}$. Then we say that $F$
solves Byzantine agreement if, for each $k$-round scenario $\sigma$ and every
process $p \in P$, the following conditions hold:
\begin{itemizedot}

  \item {\tmem{Termination}}: Every non-Byzantine process $p$ outputs value $F (\sigma_p)$.

  \item {\tmem{Validity}}: If the transmitter $p_0$ is non-Byzantine, then every
  non-Byzantine process  $p$ outputs the initial value of $p_0$,
  i.e. $F (\sigma_p) = \sigma (p_0)$ if $p, p_0 \nin B
  (\sigma)$.
  
  \item {\tmem{Agreement}}: Any two non-Byzantine processes $p$ and $q$  have the same
  output, i.e. $F (\sigma_p) = F (\sigma_q)$ if $p, q \nin B (\sigma)$.

\end{itemizedot}

\section{The Byzantine Agreement Protocol}\label{sec-algorithm}


In this section, we present an algorithm we call {BA++} (Algorithm \ref{algo-bapp}) for solving
Byzantine agreement within a (n, m, d, b)-system.
We adopt the description style of \cite{fischer1982lower} for our algorithm.
The main theorem is as follows.

\begin{theorem}\label{thm-upper-bound}
  BA++ is a $(b+3)$-round algorithm that
  solves Byzantine agreement for a (n, m, d,
  b)-system if $n > \max \{2 m + d, 2 d + m, b\}+ 2
  b$.
\end{theorem}

At a very high level (Figure \ref{fig-algorithm}), the idea underlying algorithm BA++ is the following.
The processes exchange their messages in a full information manner
during $b+3$ rounds.\footnote{We discuss how to reduce that number of rounds
in Section \ref{sec-time}.}
According to our model, the views obtained at each process contains both partial failures and Byzantine failures. The first step of BA++ is to correct the partial failures.
This is challenging because the partial faults introduced in the early rounds would
still exist in the subsequent rounds.
We address this problem by an algorithm we call \tmem{View-Transform} (Algorithm \ref{algo-view}):
this transforms a view with partial failures into a view without partial failures.
Another challenge is to ensure that the views (that resulted from a same scenario)
still belong to a same scenario after View-Transform.
This is addressed by iterations of \tmem{Local-Majority} (Algorithm \ref{algo-lm3}).
After applying View-Transform to the original view, the majority algorithm ($O M$) of Lamport
{\cite{lamport1982byzantine}} (or any $(b+1)$-round simultaneous Byzantine
agreement algorithm) can be employed to compute a output.


\begin{figure}[p]
  \centering
  \includegraphics[scale=0.25]{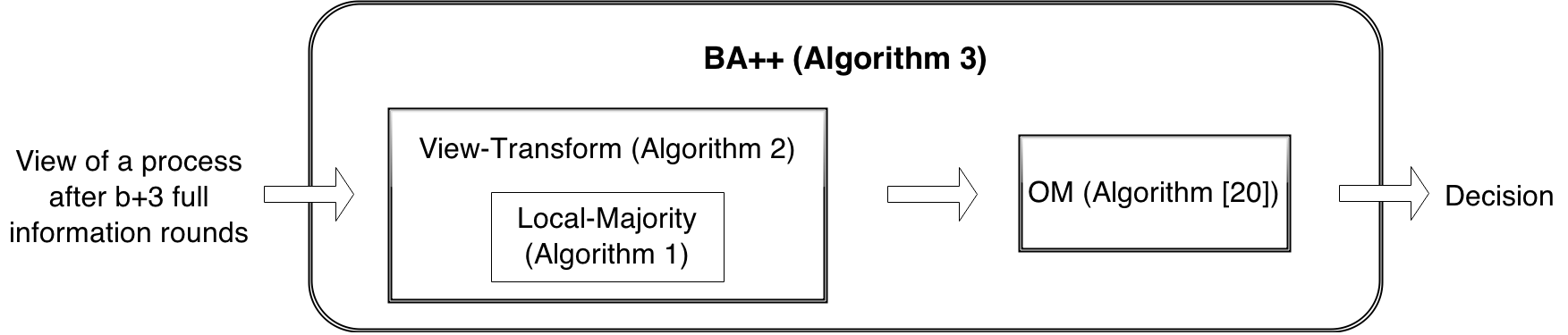}
  \caption{High-level view of Algorithm BA++\label{fig-algorithm}}
\end{figure}

\begin{algorithm}[p]\label{algo-lm3}
  {\tmstrong{Assume}}: $\sigma_p$ is a $k$-round view of process $p$ for a
  (n, m, d, b)-system with $k \geqslant 3$ and transmitter $p_0$.{\hspace*{\fill}}

  {\tmstrong{Code}} for $p$:{\hspace*{\fill}}\\
  For every string $p_0 p_1 \ldots p_i$ and string $s$ with $0 \leqslant | s |
  \leqslant k - 3 - i$:
    \begin{tmindent}
      \begin{enumerate}
        \item $p$ initializes an empty multiset $S$.
        
        \item For every process $p_{i+1} \in P \backslash p_i$, if at least $n - m - b - 1$ elements of
        $\{ \sigma^{p_0 p_1 \ldots p_i}_{s p} (p_i p_{i+1} p_{i+2}): {p_{i+2} \in P \backslash p_{i+1}} \}$ have the same value $v$,
        $p$ adds $v$ to $S$.
        
        \item If more than half of $S$ have the same value $v'$, then $p$ sets
        $L M_3 (\sigma^{p_0 p_1 \ldots p_i}_{s p})$ to $v'$. Otherwise $p$
        sets $L M_3 (\sigma^{p_0 p_1 \ldots p_i}_{s p})$ to $\bot$.
      \end{enumerate}
    \end{tmindent}
  \caption{$3$-round Local-Majority ($L M_3$)}
\end{algorithm}

\begin{algorithm}[p]\label{algo-view}
  {\tmstrong{Assume}}: $\sigma_p$ is a $k$-round view of process $p$ for a
  (n, m, d, b)-system with $k \geqslant 3$ and  transmitter $p_0$.
  $L M_3$ is Algorithm \ref{algo-lm3}.{\hspace*{\fill}}
  
  {\tmstrong{Code}} for $p$:{\hspace*{\fill}}
  
  Loop from $i = k - 3$ to $i = 0$: (denote the following $i$th iteration as transform $V
  T_i^p$.)
  \begin{tmindent}
    \begin{enumerate}
      \item Let $\sigma_p'$ be a copy of $\sigma_p$.
      
      \item $p$ changes $\sigma_p' (p_0 p_1 \ldots p_i s)$ to be $L M_3
      (\sigma^{p_0 p_1 \ldots p_i}_{s p})$ for every $p_1 \ldots p_i$ and
      every string $s$ with $0 \leqslant | s | \leqslant k - 3 - i$.
      
      \item Let $\sigma_p = \sigma_p'$. ($\sigma'_p$ is the output of $V
      T_i^p$.)
    \end{enumerate}
  \end{tmindent}
  After the loop, $p$ outputs the first $(k - 2)$-round view of
  $\sigma_p$.

  \caption{View-Transform $V T^p$ with respect to $L M_3$}
\end{algorithm}

\begin{algorithm}[p]\label{algo-bapp}
  {\tmstrong{Assume}}: $\sigma_p$ is a $(b + 3)$-round view of process $p$ for
  a (n, m, d, b)-system and transmitter $p_0$. $V T^p$ is Algorithm
  \ref{algo-view}.{\hspace*{\fill}}
  
  {\tmstrong{Code}} for $p$:{\hspace*{\fill}}
  \begin{fkindent}
  \begin{enumeratenumeric}
    \item Let $\sigma_p' = V T^p (\sigma_p)$ with respect to $L M_3$.
    
    \item Then $p$ outputs $O M (\sigma_p')$. Here $O M$ is the Byzantine
    agreement algorithm in \cite{lamport1982byzantine}.
  \end{enumeratenumeric}
  \end{fkindent}

  \caption{BA++ with respect to $L M_3$}
\end{algorithm}

\begin{lemma}\label{lemma-lm3}
  Suppose $n > \max \{ 2 m + d, 2 d + m, b \} + 2 b$.
  In $L M_3$ (Algorithm \ref{algo-lm3}),
  if $\sigma^{p_0 p_1 \ldots p_i}_{s p} (p_i p_{i+1} p_{i+2}) = \sigma^{p_0 p_1 \ldots p_i}_{s' p'} (p_i p_{i+1} p_{i+2})$ for all $p_{i+1}$ and $p_{i+2}$, then $L M_3 (\sigma^{p_0 p_1 \ldots p_i}_{s p}) = L M_3 (\sigma^{p_0 p_1 \ldots p_i}_{s' p'})$.
  If $p_i$ is non-Byzantine, then $L M_3 (\sigma^{p_0 p_1 \ldots p_i}_{p}) = \sigma (p_0 p_1 \ldots p_i)$.
\end{lemma}

\begin{proof}
  The first part of the lemma follows directly from the algorithm,
  so we only need to show the second part.

  If $m = d = 0$ and $b = 0$, the lemma follows directly since there are no
  failures. In the following, we prove the lemma in the case that $m \neq 0$ or $b
  \neq 0$.
%
  
  If $p_{i + 1}$ is correct,  then there are at
  least $n - m - b - 1$ elements in $\{
  \sigma_p (p_0 \ldots p_{i + 1} p_{i + 2}) : {p_{i + 2} \in P \backslash p_{i + 1}}\}$
  equal to $\sigma (p_0 \ldots p_{i + 1})$,
  which implies $\sigma (p_0 \ldots p_{i + 1})$ is added to $S$.
  
  If $p_{i + 1}$ is d-faulty, then there are at most
  $m + d + b - 1$ values different from $\sigma (p_0 \ldots p_{i + 1})$ in $\{
  \sigma_p (p_0 \ldots p_{i + 1} p_{i + 2}) : {p_{i + 2} \in P \backslash p_{i + 1}} \}$.
  Since $n - m - b - 1 > m + d + b - 1$, only $\sigma (p_0 \ldots p_{i
  + 1})$ might be added to $S$.
  
  Now consider $p_i$. If $p_i$ is correct,
  then all correct processes
  will contribute a value $\sigma (p_0 \ldots p_i)$ to $S$. So there
  are at least $n - 1 - m - b$ values equal to $\sigma (p_0 \ldots p_i)$ in
  $S$ and at most $b$ values in $S$ different from $\sigma (p_0 \ldots p_i)$
  (contributed by $b$ Byzantine processes).
  If $m \neq 0$, then $n > 2m + d + 2b \geqslant m+1+d+2b$.
  If $m = 0$ but $b \neq 0$, then $n > 3b \geqslant m+1+d+2b$.
  So $n - 1 - m - b$ is always greater than $b$, the majority
  value of $S$ is $\sigma (p_0 \ldots p_i)$, i.e. $L M_3 (\sigma_p) = \sigma
  (p_0 \ldots p_i)$.
  
  If $p_i$ is d-faulty, then all correct processes
  except the ones that receive wrong values from $p_i$ will contribute a value
  $\sigma (p_0 \ldots p_i)$ to $S$. So there are at least $n - m - d - b$
  values equal to $\sigma (p_0 \ldots p_i)$ in $S$, and at most $d + b$ values
  different from $\sigma (p_0 \ldots p_i)$ in $S$. Since $n > m + 2 d + 2 b$,
  the majority value of $S$ is still $\sigma (p_0 \ldots p_i)$, i.e. $L M_3
  (\sigma_p^{p_0 \ldots p_i}) = \sigma (p_0 \ldots p_i)$.
\end{proof}

We show that the output of \tmem{View-Transform} for different processes
actually comes from a single scenario of a (n, 0, 0, b)-system for which the $O M$
algorithm guarantees Byzantine agreement in $(b+1)$ rounds. We prove this by
introducing the following \tmem{Scenario-Transform}.

\floatstyle{boxed}
\restylefloat{figure}
\begin{figure}[h]
  {\tmstrong{Assume}}: {\tmstrong{}} $\sigma$ is a $k$-round scenario for a (n,
  m, d, b)-system with $k \geqslant 3$ and transmitter $p_0$. Let $B =
  B (\sigma)$ be the set of Byzantine processes. $V T_i^p$ is the $i$-th iteration  in Algorithm \ref{algo-view}. {\hspace*{\fill}}
  
  {\tmstrong{Transform}}:{\hspace*{\fill}}
  
  Loop from $i = k - 3$ to $i = 0$: (denote the following $i$th iteration as transform $S
  T_i$)
  \begin{tmindent}
    \begin{enumerate}
      \item Let $\sigma'$ be a copy of $\sigma$.
      
      \item For each $p \nin B$, apply $V T_i^p$ to $\sigma'$, i.e. $\sigma'
      (p_0 \ldots p_i s p) = V T_i^p (\sigma_p) (p_0 \ldots p_i s)$ for every
      $s \in P^{0 : k - i - 1}$. (This Line makes sense because view transforms
      are independent for different processes.)
      
      \item For every $p \nin B$, $q \in B$, $s \in P^{0 : k - i - 2}$, set
      $\sigma' (p_0 \ldots p_i s p q)$ to $\sigma' (p_0 \ldots p_i s p)$.
      
      \item Let $\sigma = \sigma'$.
    \end{enumerate}
  \end{tmindent}
  After the loop, output $\sigma$.
\caption{Scenario-Transform (ST) with respect to $L M_3$}
\label{fig-scenario}
\end{figure}

\begin{lemma}\label{lemma-scenario}
  Consider a $k$-round scenario $\sigma$ for a (n, m, d, b)-system
  with $k \geqslant 3$ and transmitter $p_0$. The output scenario
  of Scenario-Transform (Figure \ref{fig-scenario}) is a scenario of a (n, 0, 0, b)-system.
  Moreover, this output scenario satisfies $(S T(\sigma))_p = V T^p (\sigma_p)$ for any non-Byzantine process $p$. If
  $p_0$ is a non-Byzantine transmitter for $\sigma$, then $p_0$ is a correct
  transmitter for $S T (\sigma)$ such that $S T (\sigma) (p_0) = \sigma
  (p_0)$.
\end{lemma}

\begin{proof}
  For a non-Byzantine process $p$, $(S T(\sigma))_p = V T^p (\sigma_p)$ follows immediately
  from Line $2$ which uses $V T^p_i$ as in algorithm $V T^p$. We now prove that
  the output scenario is a scenario of a (n, 0, 0, b)-system.

  Let $i$th-$\sigma$ be the scenario just after the $i$th loop iteration inside $S T$.
  We prove by induction this claim: if $i \leqslant v \leqslant k - 3$, then
  $i$th-$\sigma (p_0 \ldots p_v p_{v + 1}) = i \text{th-} \sigma (p_0 \ldots
  p_v)$ for every non-Byzantine process $p_v$. Note that if $p_{v + 1} \in B$, the claim
  follows by Line $3$ of $S T$. Thus we only need to prove the claim for the case $p_{v + 1}$ is non-Byzantine.
  
  First consider $i = k - 3$. In this case, $v$ could only be $k - 3$. Suppose
  $p_{k - 3}$ and $p_{k - 2}$ are non-Byzantine. Then $(k - 3) \text{th-} \sigma (p_0
  \ldots p_{k - 3} p_{k - 2}) = V T_{k - 3}^{p_{k - 2}} (\sigma_{p_{k - 2}}) (p_0
  \ldots p_{k - 3})$. According to Line $2$ of $V T_{k - 3}^{p_{k - 2}}$, $V
  T_{k - 3}^{p_{k - 2}} (\sigma_{p_{k - 2}}) (p_0 \ldots p_{k - 3}) = L M_3
  (\sigma_{p_{k - 2}}^{p_0 \ldots p_{k - 3}})$. Since $p_{k - 3} \nin B$, by
  Lemma \ref{lemma-lm3} $L M_3 (\sigma_{p_{k - 2}}^{p_0
  \ldots p_{k - 3}}) = \sigma^{p_0 \ldots p_{k - 3}} (p_{k - 3}) = \sigma (p_0
  \ldots p_{k - 3})$. Since $(k - 3) \text{th-} \sigma (p_0 \ldots p_{k - 3}) =
  \sigma (p_0 \ldots p_{k - 3})$, the claim for $k - 3$ is proved.
  
  Now suppose the claim is true for $i + 1$. Let us prove it for $i$.
  We need to show the claim for all $i \leqslant v \leqslant k - j$.
  First, consider $v = i$ and suppose $p_i$ and $p_{i + 1}$ are non-Byzantine. Then
  according to Line $2$ of $V T_i^{p_{i + 1}}$, $i$th-$\sigma (p_0 \ldots p_i
  p_{i + 1}) = L M_3 \left( (i + 1) \text{th-} \sigma^{p_0 \ldots p_i}_{p_{i
  + 1}} \right)$. Since $p_i$ is non-Byzantine, according to Lemma \ref{lemma-lm3}, $L M_3
  \left( (i + 1) \text{th-} \sigma^{p_0 \ldots p_i}_{p_{i + 1}} \right) = (i
  + 1) \text{th-} \sigma^{p_0 \ldots p_i} (p_i) = (i + 1) \text{th-}
  \sigma (p_0 \ldots p_i)$.
  Hence, $i\text{th-}\sigma (p_0 \ldots p_i p_{i + 1}) = (i + 1) \text{th-} \sigma (p_0 \ldots p_i)$.
  Because the value for $p_0 \ldots p_i$ is not
  changed in the $i$th loop of $S T$, $i \text{th-} \sigma (p_0 \ldots p_i) = (i
  + 1) \text{th-} \sigma (p_0 \ldots p_i)$.
  Thus $i \text{th-} \sigma (p_0 \ldots p_i) = i \text{th-} \sigma (p_0 \ldots p_i p_{i + 1})$, the claim is true
  for $v = i$. Now consider $v > i$. According to $V T_i^{p_{v+1}}$ and $V
  T_i^{p_v}$, $i \text{th-} \sigma (p_0 \ldots p_v p_{v + 1}) = L M_3 \left(
  (i + 1) \text{th-} \sigma^{p_0 \ldots p_i}_{p_{i + 1} \ldots p_{v + 1}}
  \right)$ and $i \text{th-} \sigma (p_0 \ldots p_v) = L M_3 \left( (i + 1)
  \text{th-} \sigma^{p_0 \ldots p_i}_{p_{i + 1} \ldots p_v} \right)$. Since
  $p_v$ is correct and $v > i$, by induction hypothesis $(i + 1) \text{th-}
  \sigma^{p_0 \ldots p_i}_{p_{i + 1} \ldots p_{v + 1}}$ is equal to $(i + 1)
  \text{th-} \sigma^{p_0 \ldots p_i}_{p_{i + 1} \ldots p_v}$. Therefore
  $i$th-$\sigma (p_0 \ldots p_v p_{v + 1}) = i \text{th-} \sigma (p_0 \ldots
  p_v)$, and the claim is proved.
  
  From the claim, we see that in $S T (\sigma)$ every non-Byzantine process
  always sends correct messages to other processes.
  So $S T (\sigma)$ is a scenario of
  (n, 0, 0, b)-system with Byzantine processes $B(\sigma)$. Therefore, if $p_0$ is
  non-Byzantine in $\sigma$ then $p_0$ is also correct in $S T (\sigma)$. Because
  the value of $\sigma (p_0)$ for non-Byzantine process $p_0$ is never changed in $S T$, \
  $S T (\sigma) (p_0) = \sigma (p_0)$.
\end{proof}

With all the lemmas above, now we can give a proof of Theorem \ref{thm-upper-bound}.

\begin{proof}[Proof of Theorem \ref{thm-upper-bound}]
  Suppose $\sigma$ is a $(b + 3)$-round scenario for (n, m, d, b)-system. By
  Lemma \ref{lemma-scenario} above, $\sigma' = S T (\sigma)$ with respect to $L M_3$ is a $(b +
  1)$-round scenario of (n, 0, 0, d)-system. Since $V T^p (\sigma_p) =
  \sigma'_p$ for every non-Byzantine process $p$, $O M (V T^p (\sigma_p))$ are equal
  for all non-Byzantine process which proves the agreement property. Moreover,
  if $p_0$ is non-Byzantine, then $O M (V T^p (\sigma_p)) = S T (\sigma) (p_0)$.
  This shows the validity property. Therefore, the theorem is proved.
\end{proof}

\section{Resilience Lower Bounds}\label{sec-resilience}

We show here that our BA++ algorithm is optimal  with
respect to resilience; namely, $n > \max \{2 m + d, 2 d + m, b\}+ 2 b$ is a
tight bound to reach Byzantine agreement. If $m = d = 0$, this bound is $n > 3 b$ which is tight by
\cite{lamport1982byzantine}. So in this section, we assume that $m, d > 0$ and
show that it is
impossible to achieve Byzantine agreement if $n \leqslant 2 m + d + 2 b$ or $n \leqslant 2 d + m + 2 b$.

\begin{lemma}\label{lemma-impsb1}
  If $n \leqslant 2 m + d + 2 b$, then there is no Byzantine
  agreement algorithm in a (n, m, d, b)-system.
\end{lemma}

\begin{proof}
  Consider a Byzantine agreement algorithm $F$ for a
  (n, m, d, b)-system. Since $n \leqslant 2 m + d + 2 b$, $P$ can be
  partitioned into five non-empty sets $G$, $H$, $I$, $J$ and $K$, with $| G |
  \leqslant m$, $| H | \leqslant m$, $| I | \leqslant b$, $| J | \leqslant b$,
  $| K | \leqslant d$. Select an arbitrary process in $G$ as transmitter
  $p_0$. We define scenarios $\alpha$ and $\beta$ recursively as follows:
  \smallskip
  \begin{enumerateroman}
    \item For every $p \in P$, $k \in K$, $q \in P \backslash K$, let
    \[ \alpha (p_0) = 0, \alpha (p_0 p) = 0, \]
    \[ \beta (p_0) = 1, \beta (p_0 k) = 0, \beta (p_0 q) = 1, \]
    \item For every $g \in G$, $h \in H$, $i \in I$, $j \in J$, $k \in K$, $p
    \in P$, $q \in P \backslash K$, $w \in p_0 P^{\ast}$, define the following values
    recursively on the length of $w$:
    \[ \alpha (w g p) = \alpha (w g) \nocomma, \alpha (w i p) = \alpha (w i),
       \alpha (w k p) = \alpha (w k), \]
    \[ \beta (w h p) = \beta (w h) \nocomma, \beta (w j p) = \beta (w j),
       \beta (w k p) = \beta (w p), \]
    \[ \alpha (w h k) = \beta (w h k) \nocomma, \alpha (w h q) = \alpha (w h),
       \alpha (w j p) = \beta (w j p) \nocomma, \text{} \]
    \[ \beta (w g k) = \alpha (w g k), \beta (w g q) = \beta (w g), \beta (w i
       p) = \alpha (w i p) . \]
  \end{enumerateroman}
  It is easy to check that $\alpha$ is a scenario of a (n,m,d,b)-system with d-faulty processes in $H$ and Byzantine processes in $J$, and that $\beta$ is a scenario of a (n,m,d,b)-system with d-faulty processes in $G$ and Byzantine processes in $I$.

  In the construction, $\alpha_k = \beta_k$ for all $k \in K$. Thus, $F (\alpha_k) = F (\beta_k)$ for all $k \in K$.
  Since $p_0$ is a non-Byzantine process in both $\alpha$ and $\beta$, according to Byzantine agreement we have
  \[ F (\alpha_k) = \alpha (p_0) = 0, \]
  \[ F (\beta_k) = \beta (p_0) = 1. \]
  However, it is a contradiction to that $F (\alpha_k) = F (\beta_k)$ for all $k \in K$.
  The lemma is proved.
\end{proof}

\begin{lemma}\label{lemma-impsb2}
  If $n \leqslant 2 d + m + 2 b$, then there is no Byzantine
  agreement algorithm in a (n, m, d, b)-system.
\end{lemma}

The proof for Lemma \ref{lemma-impsb2} is similar to the proof of Lemma \ref{lemma-impsb1}.
Due to space limitation, we defer the proof to Appendix \ref{app-omitted}.

Taking together the algorithm in Section \ref{sec-algorithm} and the lemmas above, 
we have the following theorem.

\begin{theorem}\label{thm-bound}
  Byzantine agreement can be solved in a (n, m, d, b)-system if and only if $n > \max \{2 m + d, 2 d + m, b\}+ 2
  b$.
\end{theorem}

\subsection*{Signed messages}\label{sec-signed}

So far we have assumed oral message.
We now discuss the case where processes could send signed messages {\cite{lamport1982byzantine}}. 
In this case, we also have a tight bound on the number of processes
for reaching Byzantine agreement.
Following {\cite{lamport1982byzantine}}, a signed message satisfies the
following two properties:
\begin{enumerate}
\itemsep0em
  \item The signature of a non-Byzantine process cannot be forged and any alteration of
  its content can be detected.
  
  \item Every process can verify the authenticity of a signature.
\end{enumerate}

Formally, suppose $\sigma$ is a $k$-round scenario for a (n, m, d, b)-system
with signed messages. Let $\sigma (p_0 p_1 \ldots p_i p)$ ($i < k$) be a
message received by process $p$. If process $p_j$ ($j \leqslant i$) is
non-Byzantine, then either $\sigma (p_0 \ldots p_i p) = \sigma (p_0 \ldots p_j)$, or the signature of $p_j$ is forged.


\begin{algorithm}[h]
    {\tmstrong{Assume}}: $\sigma_p$ is a $(b + 2)$-round view of process $p$
    for a (n, m, d, b)-system with signed messages, and $p_0$ is the
    transmitter.
    
    {\tmstrong{Code}} for $p$:
    \begin{fkindent}
      \begin{enumeratenumeric}
        \item $p$ initializes an empty set $S$.
        
        \item For every string $p_0 \ldots p_i$ ($0 \leqslant i \leqslant b +
        1$, and $p_0, \ldots, p_i$ are different processes): if the signatures
        attached to value $\sigma_p (p_0 \ldots p_i)$ are correct, then $p$
        adds $\sigma (p_0 p_1 \ldots p_i)$ into $S$.
        
        \item $p$ outputs the majority value of $S$.
      \end{enumeratenumeric}
    \end{fkindent}
    \caption{Algorithm $\tmop{SBA}$++\label{algo-sba1}}
\end{algorithm}

We present an algorithm called SBA++ (Algorithm \ref{algo-sba1}) for solving Byzantine agreement with signed message.
Due to space limitation, we move the proof of Algorithm SBA++ and the following theorem into Appendix \ref{app-sba}.

\begin{theorem}
  \label{thm-sba}Byzantine agreement can be solved for a (n, m, d, b)-system
  with signed messages if and only if $n > m + d + b$.
\end{theorem}

\section{Time Optimality}\label{sec-time}

In this section, we investigate the time complexity of reaching Byzantine
agreement for a (n, m, d, b)-system. If $m = 0$, the communication rounds
needed to reach Byzantine agreement is $b + 1$ by
\cite{fischer1982lower}. So in this section, we assume $m > 0$. We show that
in some cases ($n \geqslant \max \{ 2 m + 2 d, b + 1 \} + 2 b$) the lower
bound of the number of rounds for reaching Byzantine agreement is $b + 2$, and in other cases
(e.g. $b = 0$) the lower bound is $b + 3$.

We first show that a ($b + 2$)-round algorithm is available if $n \geqslant
\max \{ 2 m + 2 d, b + 1 \} + 2 b$. In this case we have the following $2$-round
Local-Majority algorithm.

\begin{algorithm}[H]
  {\tmstrong{Assume}}: $\sigma_p$ is a $k$-round view of process $p$ for a
  (n, m, d, b)-system with $k \geqslant 3$ and $p_0$ is the
  transmitter.{\hspace*{\fill}}

  {\tmstrong{Code}} for $p$:{\hspace*{\fill}}\\
  For every string $p_0 p_1 \ldots p_i$ and string $s$ with $0 \leqslant | s |
  \leqslant k - 3 - i$:
    \begin{tmindent}
      \begin{enumerate}
        \item If more than half of $\{ \sigma^{p_0 p_1 \ldots p_i}_{s p} (p_i p_{i+1}) : {p_{i+1} \in P \backslash p_{i}} \}$ have the same value $v$, then $p$ sets
        $L M_2 (\sigma^{p_0 p_1 \ldots p_i}_{s p})$ to $v$. Otherwise $p$
        sets $L M_2 (\sigma^{p_0 p_1 \ldots p_i}_{s p})$ to $\bot$.
      \end{enumerate}
    \end{tmindent}
  \caption{$2$-round Local-Majority ($L M_2$)\label{algo-lm2}}
\end{algorithm}

\begin{lemma}\label{lemma-lm2}
  Suppose $n \geqslant 2 m + 2 d + 2 b$ and $n > 2b + 1$.
  In $L M_2$ (Algorithm \ref{algo-lm2}),
  if $\sigma^{p_0 p_1 \ldots p_i}_{s p} (p_i p_{i+1}) = \sigma^{p_0 p_1 \ldots p_i}_{s' p'} (p_i p_{i+1})$ for all $p_{i+1}$, then $L M_2 (\sigma^{p_0 p_1 \ldots p_i}_{s p}) = L M_2 (\sigma^{p_0 p_1 \ldots p_i}_{s' p'})$.
  If $p_i$ is non-Byzantine, then $L M_2 (\sigma^{p_0 p_1 \ldots p_i}_{p}) = \sigma (p_0 p_1 \ldots p_i)$.
\end{lemma}

\begin{proof}
  The first part of the lemma follows directly from the algorithm.
  So we only need to show the second part.

  If $p_i$ is correct, then in $\{ \sigma_p (p_0 \ldots
  p_{i + 1}) : {p_{i + 1} \in P \backslash p_i} \}$ there are at least $n - 1 - m - b$
  values equal to $\sigma (p_0 \ldots p_i)$ and at most $m + b$ values
  different from $\sigma (p_0 \ldots p_i)$ of which $b$ values are contributed
  by $B (\sigma)$ and $m$ values are contributed by $D (\sigma)$.
  If $m  = d = 0$, then $n > 2b+1 = 2m+2b+1$. If $m \neq 0$ and $d \neq 0$,
  then $n \geqslant 2m + 2d +2b > 2m + 2b + 1$.
  So $n$ is always greater than
  $m + 2 b + 1$, the majority values of $\{ \sigma_p (p_0 \ldots
  p_{i + 1}) : {p_{i + 1} \in P \backslash p_i} \}$ are equal to $\sigma (p_0 \ldots p_i)$,
  i.e. $L M_2 (\sigma^{p_0 p_1 \ldots p_i}_{p}) = \sigma (p_0 p_1 \ldots p_i)$.
  
  If $p_i$ is d-faulty, then in $\{ \sigma_p (p_0 \ldots
  p_{i + 1}) : {p_{i + 1} \in P \backslash p_i} \}$ there are at least $n - m - d - b$
  values equal to $\sigma (p_0 \ldots p_i)$ and at most $m - 1 + d + b$ values
  different from $\sigma (p_0 \ldots p_i)$ of which $b$ values are contributed
  by $B (\sigma)$ and $m - 1 + d$ values are contributed by $D (\sigma)$.
  Since $n \geqslant 2 m + 2 d + 2 b$, we have $n - 1 > 2 (m - 1 + d + b)$,
  the majority values are equal to $\sigma (p_0 \ldots p_i)$,
  i.e. $L M_2 (\sigma^{p_0 p_1 \ldots p_i}_{p}) = \sigma (p_0 p_1 \ldots p_i)$.
\end{proof}

\begin{lemma}
  If $n \geqslant \max \{ 2 m + 2 d, b + 1 \} + 2 b$, then Byzantine agreement
  can be solved in $b + 2$ rounds for a (n, m, d, b)-system.
\end{lemma}

\begin{proof}
  Section \ref{sec-algorithm} uses $3$-round algorithm $L M_3$ 
  (Algorithm \ref{algo-lm3}) to implement the sub-algorithm
  \tmem{View-Transform} (Algorithm \ref{algo-view}), and then get a ($b+3$)-round Byzantine agreement
  algorithm. When $n \geqslant \max \{ 2 m + 2 d, b + 1 \} + 2 b$, we have a $2$-round
  Local-Majority algorithm $L M_2$. Thus if we replace $L M_3$ with $L M_2$, we 
  obtain a ($b+2$)-round Byzantine agreement algorithm.
\end{proof}

In the following, we prove that $b + 2$ is also a lower bound of rounds for reaching
Byzantine agreement. Specially, $b + 2$ is a tight bound for the case
$n \geqslant \max \{ 2 m + 2 d, b + 1 \} + 2 b$.

\begin{theorem}
  Byzantine agreement for a (n, m, d, b)-system ($m,
  d > 0$) requires at least $b + 2$ rounds.
\end{theorem}

\begin{proof}
  Suppose in contrary that there is a $(b + 1)$-round Byzantine agreement
  algorithm $F$. For any string $w$, we use $\bar{w}$ to denote the number
  corresponding to $w$ with radix $n$.
  
  Select an arbitrarily process $p_0$ in the system as a fixed transmitter.
  For $0 \leqslant x \leqslant n^{b + 1} + 1$, define $\alpha_x : p_0 P^{0 :
  b} \rightarrow \{ 0, 1 \}$ as
  \[ \tmop{for} w \in p_0 P^{0 : b} \nocomma, \alpha_x (w) = \left\{
     \begin{array}{ll}
       0 & \tmop{if} \overline{w} < x,\\
       1 & \tmop{otherwise.}
     \end{array} \right. \]
  It is easy to see that $\alpha_0 (w)$ is always equal to $1$, so $F
  (\alpha_0) = 1$. For the same reason, $F (\alpha_{n^{b + 1} + 1}) = 0$. We
  claim: $\alpha_x$ and $\alpha_{x + 1}$ are views derived from a same
  scenario for all $1 \leqslant x \leqslant n^{b + 1}$. If so, by the
  agreement property of $F$ we have $F (\alpha_x) = F (\alpha_{x + 1})$. Then,
  we have $F (\alpha_0) = F (\alpha_1) = \ldots = F (\alpha_{n^{b + 1} + 1})$.
  This is a contradiction to $F (\alpha_0) = 1$ and $F (\alpha_{n^{b + 1} +
  1}) = 0$. Now it remains to prove the claim.
  
  For $1 \leqslant x \leqslant n^{b + 1}$, let $x = \overline{q_0 q_1 \ldots
  q_b}$. Since $n > b + 3$, there exists two different processes $q_{b + 1}$
  and $q_{b + 2}$ (assume $q_{b + 1} > q_{b + 2}$ without loss of generality)
  in $P \backslash \{ q_0 \ldots q_b \}$. Define a function $\sigma : p_0 P^{0 : b +
  1} \rightarrow \{ 0, 1 \}$ as
  \[ \tmop{for} w \in p_0 P^{0 : b + 1} \nocomma, \sigma (w) = \left\{
     \begin{array}{ll}
       0 & \tmop{if} p_0 < q_0,\\
       0 & \tmop{if} p_0 = q_0 \tmop{and} w = q_0 \ldots q_i q s, \tmop{with}
       0 \leqslant i \leqslant b, q < q_{i + 1},\\
       1 & \tmop{otherwise.}
     \end{array} \right. \]
  It is easy to check that $\sigma_{q_{b + 1}} = \alpha_x$ and $\sigma_{q_{b +
  2}} = \alpha_{x + 1}$. If $p_0 < q_0$, then $\sigma (w)$ is always equal to
  $0$. So $\alpha_x$ and $\alpha_{x + 1}$ come from an admissible scenario
  $\sigma$. If $p_0 > q_0$, for the similar reason the claim is correct. If
  $q_0 = p_0$, then for every process $p$ in $P \backslash \{ q_0, \ldots, q_b \}$ we
  always have $\sigma (w p q) = \sigma (w p)$. If the set $\{
  q_0, \ldots, q_b \}$ has less than $b$ elements, then let $B (\sigma) = \{
  q_0, \ldots, q_b \}$ and $\sigma$ is a $(b + 1)$-round scenario. Thus
  $\alpha_x$ and $\alpha_{x + 1}$ come from an admissible scenario $\sigma$.
  If the set $\{ q_0, \ldots, q_b \}$ has $b + 1$ different elements, then let
  $\phi$ be as follows:
  \[ \tmop{for} w \in p_0 P^{0 : b + 1} \nocomma, \phi (w) = \left\{
     \begin{array}{ll}
       1 & \tmop{if} w = q_0 \ldots q_b q \tmop{with} q < q_{b + 1} \tmop{and}
       q \neq q_{b + 2},\\
       \sigma (w) & \tmop{otherwise.}
     \end{array} \right. \]
  $\phi$ is a $(b + 1)$-round scenario with Byzantine processes $\{ q_0,
  \ldots, q_{b - 1} \}$ and d-faulty processes $\{ q_b \}$. Also we have
  $\phi_{q_{b + 1}} = \sigma_{q_{b + 1}} = \alpha_x$ and $\phi_{q_{b + 2}} =
  \sigma_{q_{b + 2}} = \alpha_{x + 1}$. Thus $\alpha_x$ and $\alpha_{x + 1}$
  come from an admissible scenario $\phi$. Hence, the claim we mentioned is
  always correct. So the theorem follows.
\end{proof}

Now we show that $b + 3$ could be lower bound in certain cases. Specifically,
suppose $b = 0$, we prove that $3$ rounds is a lower bound.
\begin{lemma}
Suppose $m, d > 0$ and $max \{2 m + d, 2 d + m\} < n < 2 m + 2 d$, then there is no
$2$-round Byzantine agreement algorithm for a (n, m, d, 0)-system.
\end{lemma}

\begin{proof}
  Let $F$ be a $2$-round Byzantine agreement algorithm. Select an arbitrarily
  process $p_0$ in $P$ as transmitter. By the assumption of the lemma, $P \backslash
  p_0$ can be partitioned into four sets $G$, $H$, $I$ and $J$ such
  that $| G | \leqslant m - 1$, $| H | \leqslant m - 1$, $0 < | I | \leqslant d$,
  $0 < | J | \leqslant d$.
  We define two $2$-round scenarios $\alpha$ (with d-faulty processes in $G \cup \{ p_0 \}$) and $\beta$ (with 
  d-faulty processes in $H \cup \{ p_0 \}$)
  as follows.
  \begin{enumerateroman}
    \item For every $i \in I$, $j \in J$, $q_i \in P \backslash I$, $q_j \in P \backslash J$ let
    \[ \alpha (p_0) = 0, \alpha (p_0 i) = 1, \alpha (p_0 q_i) = \alpha (p_0),
    \]
    \[ \beta (p_0) = 1, \beta (p_0 j) = 0, \beta (p_0 q_j) = \beta (p_0), \]
    \item For every $g \in G $, $h \in H $, $i \in I$, $q_g
    \in P \backslash (G \cup \{ p_0 \})$, $q_h \in P \backslash (H \cup \{ p_0 \})$, $q_i \in P
    \backslash I$, $p \in P$ let
  \[ \alpha (p_0 p_0 i) = \beta (p_0 p_0 i) = 1, \]
  \[ \alpha (p_0 g i) = 1, \alpha (p_0 g q_i) = \alpha (p_0 g), \alpha (p_0
     q_g p) = \alpha (p_0 q_g), \]
  \[ \beta (p_0 h i) = 0, \beta (p_0 h q_i) = \beta (p_0 h), \beta (p_0 q_h p)
     = \beta (p_0 q_h) . \]
  \end{enumerateroman}
  In the construction, $\alpha_i = \beta_i$ for all $i \in I$. Thus for any $i \in I$,
  \[ 0 = \alpha (p_0) = F (\alpha_i) = F (\beta_i) = \beta (p_0) = 1, \]
  giving a contradiction.
\end{proof}


\section{Concluding Remarks}\label{sec-rel}

There have been several attempts  to overcome the need for three-times redundancy 
in Byzantine agreement
{\cite{rabin1989verifiable,fitzi2000partial,ravikant2004byzantine,considine2005byzantine,jaffe2012price}}.
Several researchers considered stronger communication models such as
broadcast channels. In the synchronous setting, Rabin and Ben-Or
{\cite{rabin1989verifiable}} introduced the notion of global broadcast
channel and showed that any multiparty computation could be achieved with
two-times redundancy only. A partial broadcast channel was defined by Fitzi
and Maurer {\cite{fitzi2000partial}}, and corresponding lower bounds for
reaching Byzantine agreement were presented in 
{\cite{ravikant2004byzantine,considine2005byzantine,jaffe2012price}}.
Problems of secure communication and
computation in the presence of a Byzantine adversary within an {\cite{dolev1982byzantine,franklin1998secure}} incomplete network
have also been studied {\cite{dolev1982byzantine,franklin1998secure}}.

Accounting for the fact that communication failures sometimes dominate
computation ones (due to the high reliability of hardware and operating
systems), some models focused on communication failures
{\cite{perry1986distributed,schmid2002formally}} or hybrid failures
{\cite{gong1998byzantine,lincoln1993formally}}. These include models where 
the Byzantine components are the communication channels
instead of (or in addition to) the processes. For instance, in
{\cite{santoro1989time,santoro2007agreement}}, Santoro and Widmayer showed
that agreement cannot be achieved with $\left\lceil \frac{n-1}{2} \right\rceil$
Byzantine communication faults. 
Our Theorem \ref{thm-bound}
generalizes this result. Actually in Theorem \ref{thm-bound},
taking $m = \left\lceil \frac{n-1}{2} \right\rceil$, $d = 1$ and $b = 0$ would
force $n < 2 m + d + b$, which implies the impossibility of Byzantine
agreement.

\noindent
\newpage
\footnotesize

\bibliographystyle{unsrt}
\bibliography{ref}

\newpage
\normalsize
\appendices

\section{Proof of Lemma \ref{lemma-impsb2}}\label{app-omitted}

\begin{proof}
  Consider a Byzantine agreement algorithm $F$ for a
  (n, m, d, b)-system. Since $n \leqslant 2 d + m + 2 b$, $P$ can be
  partitioned into five non-empty sets $G$, $H$, $I$, $J$ and $K$, with $| G |
  \leqslant m$, $| H | \leqslant d$, $| I | \leqslant d$, $| J | \leqslant b$,
  $| K | \leqslant b$. Select an arbitrarily process in $G$ as transmitter
  $p_0$. We define scenarios $\alpha$ and $\beta$ recursively as follows:
  \begin{enumerateroman}
    \item For every $h \in H$, $i \in I$, $q_{\alpha} \in P \backslash H$, $q_{\beta}
    \in P \backslash I$ let
    \[ \alpha (p_0) = 0, \alpha (p_0 h) = 1, \alpha (p_0 q_{\alpha}) = 0, \]
    \[ \beta (p_0) = 1, \beta (p_0 i) = 0, \beta (p_0 q_{\beta}) = 1, \]
    \item For every $g \in G$, $h \in H$, $i \in I$, $j \in J$, $k \in K$, $p
    \in P$, $q_{\alpha} \in P \backslash H$, $q_{\beta} \in P \backslash I$, $w \in p_0 P^{\ast}$, define the following values recursively on the
    length of $w$:
    \[ \alpha (w h p) = \alpha (w g) \nocomma, \alpha (w i p) = \alpha (w i),
       \alpha (w k p) = \alpha (w k), \]
    \[ \beta (w h p) = \beta (w h) \nocomma, \beta (w i p) = \beta (w i),
       \beta (w j p) = \beta (w j), \]
    \[ \alpha (w g q_{\alpha}) = \alpha (w g), \alpha (w g h) = \beta (w g),
       \alpha (w j p) = \beta (w j p), \]
    \[ \beta (w g q_{\beta}) = \beta (w g), \beta (w g i) = \alpha (w g),
       \beta (w k p) = \alpha (w k p) . \]
  \end{enumerateroman}
  It is easy to check that $\alpha$ is a scenario of a (n,m,d,b)-system with d-faulty processes in $G$ and Byzantine processes in $J$, and that $\beta$ is a scenario of a (n,m,d,b)-system with d-faulty processes in $G$ and Byzantine processes in $K$.

  In the construction, $\alpha_h =
  \beta_h$ and $\alpha_i = \beta_i$ for all $h \in H$ and $i \in I$. Thus for
  any $h \in H$,
  \[ 0 = \alpha (p_0) = F (\alpha_h) = F (\beta_h) = \beta (p_0)
     = 1, \]
  giving a contradiction.
\end{proof}

\section{Byzantine Agreement with Signed Messages\label{app-sba}}

We consider that processes send signed messages. Following
{\cite{lamport1982byzantine}}, a {\tmem{signed message}} satisfies the
following two properties:
\begin{enumerate}
  \item A non-Byzantine process's signature cannot be forged and any alteration of
  the content of its signed messages can be detected.
  
  \item Any process can verify the authenticity of a process's signature.
\end{enumerate}
Formally, suppose $\sigma$ is a $k$-round scenario for a (n, m, d, b)-system
with signed messages. Let $\sigma (p_0 p_1 \ldots p_i p)$ ($i < k$) be a
message received by process $p$. If process $p_j$ ($j \leqslant i$) is
non-Byzantine, then either
\begin{enumeratenumeric}
  \item $\sigma (p_0 \ldots p_i p) = \sigma (p_0 \ldots p_j)$, or
  
  \item the signature of $p_j$ is forged.
\end{enumeratenumeric}
In this new setting, we have the following main result:

\begin{theorem}
  \label{thm-app-sba}Byzantine agreement can be solved for a (n, m, d, b)-system
  with signed messages if and only if $n > m + d + b$.
\end{theorem}

\begin{algorithm}[h]
    {\tmstrong{Assume}}: $\sigma_p$ is a $(b + 2)$-round view of process $p$
    for a (n, m, d, b)-system with signed messages, and $p_0$ is the
    transmitter.
    
    {\tmstrong{Code}} for $p$:
    \begin{fkindent}
      \begin{enumeratenumeric}
        \item $p$ initializes an empty set $S$.
        
        \item For every string $p_0 \ldots p_i$ ($0 \leqslant i \leqslant b +
        1$, and $p_0, \ldots, p_i$ are different processes): if the signatures
        attached to value $\sigma_p (p_0 \ldots p_i)$ are correct, then $p$
        adds $\sigma (p_0 p_1 \ldots p_i)$ into $S$.
        
        \item $p$ outputs the majority value of $S$.
      \end{enumeratenumeric}
    \end{fkindent}
    \caption{Algorithm $\tmop{SBA}$++\label{algo-sba}}
\end{algorithm}

\begin{lemma}
  $\tmop{SBA}$++ (Algorithm \ref{algo-sba}) solves Byzantine agreement for a (n,
  m, d, b)-system with signed messages if $n > m + d + b$.
\end{lemma}

\begin{proof}
  First suppose the transmitter $p_0$ is non-Byzantine. By definition of a signed
  message, every message $\sigma_p (p_0 \ldots p_i)$ ($i \leqslant b + 1$) is
  either equal to $\sigma (p_0)$, or is detected as forged message. So set
  $S$ contains at most $\sigma (p_0)$. If $p_0$ is correct, then
  $\sigma_p (p_0) = \sigma (p_0)$ and $\sigma (p_0)$ is added into $S$. If
  $p_0$ is d-faulty, then there must be at least one correct process such that $\sigma (p_0 q) = \sigma (p_0)$
  since $n > m + d + b$. And then we have $\sigma (p_0 q p) = \sigma (p_0 q) =
  \sigma (p_0)$. According to Line $2$ of SBA++, $\sigma_p (p_0 q) =
  \sigma (p_0)$ is added into $S$. Therefore, $S$ contains a single value
  $\sigma (p_0)$. Consequently, if $p$ is non-Byzantine, $p$ outputs $\sigma (p_0)$.
  
  Now assume the transmitter $p_0$ is Byzantine.
  Let $S_p$ and $S_{p'}$ be the corresponding set $S$ initiated by $p$ and $p'$
  in Line $1$ of SBA++.
  We show that $S_p = S_{p'}$
  for any two non-Byzantine processes $p$ and $p'$. Suppose $\sigma_p (p_0 \ldots
  p_k)$ is included in $S_p$. Let $p_l$ ($l \leqslant k$) be the non-Byzantine
  process in $p_0 \ldots p_k$ with the smallest subscript. Then all processes
  in $p_0 \ldots p_{l - 1}$ are Byzantine, which implies $l \leqslant b$. Since
  the signatures attached to $\sigma_p (p_0 \ldots p_k)$ are correct, $\sigma
  (p_0 \ldots p_l) = \sigma_p (p_0 \ldots p_k)$. Since $n - m - b > d$, $p_l$
  sends $\sigma (p_0 \ldots p_l)$ to at least one correct process $q$. Then $\sigma_{p'} (p_0 \ldots p_l q)$
  is equal to $\sigma (p_0 \ldots p_l)$. According to Line $2$ of
  SBA++, $\sigma_{p'} (p_0 \ldots p_l q) = \sigma (p_0 \ldots p_l) =
  \sigma_p (p_0 \ldots p_k)$ is added into $S_{p'}$. Therefore, we have $S_p
  \subset S_{p'}$. Since $p$ and $p'$ are two arbitrary non-Byzantine process, we
  also have $S_{p'} \subset S_p$. That is to say $S_p = S_{p'}$. According to
  Line $3$ of SBA++, all non-Byzantine processes output a same value.
\end{proof}

\begin{proof}[Proof of Theorem \ref{thm-app-sba}]
  From the lemma above, we know that if $n > m + d + b$ then Byzantine
  agreement is solvable. Now we show that if $n \leqslant m + d + b$ then Byzantine
  agreement is impossible.
  
  Suppose by contradiction that $F$ is a Byzantine agreement algorithm for a
  (n, m, d, b)-system with signed messages and $n \leqslant m + d + b$. We separate
  the processes into three sets $G$, $H$ and $I$ such that $| G | \leqslant m$, $| H |
  \leqslant b$, $| I | \leqslant d$. Select an arbitrarily process in $G$ as transmitter $p_0$.
  We define the scenarios $\alpha$ and $\beta$ (both with Byzantine processes in $H$ and d-faulty processes in $G$) recursively as follows:
  \begin{enumerateroman}
    \item For every $i \in I$, $q \in P \backslash I$ let
    \[ \alpha (p_0) = 0, \alpha (p_0 i) = \bot, \alpha (p_0 q) = \alpha (p_0),
    \]
    \[ \beta (p_0) = 1, \beta (p_0 i) = \bot, \beta (p_0 q) = \beta (p_0) . \]
    \item For every $g \in G$, $h \in H$, $i \in I$, $p \in P$, $q \in P \backslash I$,
    $w \in p_0 P^{\ast}$, define the following values recursively on the
    length of $w$:
    \[ \alpha (w g i) = \beta (w g i) = \bot \nocomma, \alpha (w g q) = \alpha
       (w g) \nocomma, \beta (w g q) = \beta (w g), \]
    \[ \alpha (w h p) = \alpha (w i p) = \beta (w h p) = \beta (w i p) = \bot
       . \]
  \end{enumerateroman}
  Moreover, $\alpha_i = \beta_i$ for all $i \in I$ since $\alpha_i (w) =
  \beta_i (w) = \bot$ for all string $w \in p_0 P^{\ast}$. Thus for any $i \in
  I$,
  \[ 0 = \alpha (p_0) = F (\alpha_i) = F (\beta_i) = \beta (p_0) = 1, \]
  giving a contradiction.
\end{proof}

\section{Early Decision}\label{app-early}

The work in \cite{dolev1982polynomial,krings1999byzantine} showed that processes
could make an early decision if the number of actual Byzantine failures is less than the
maximal number of failures it can tolerate. We show here how we can achieve early
decision with partial Byzantine failures.

\begin{theorem}
Consider a (n, m, d, b)-system ($m, d > 0$) and $f$ denotes the number of actual
Byzantine processes during an execution. Then Byzantine agreement can be solved in
the following number of rounds:
\begin{itemize}
\item $min \{ 2(f + 2), 2(d + 1) \}$, if $n \geqslant max\{2m+2d, b + 1\} + 2b$,
\item $min \{ 3(f + 2), 3(d + 1) \}$, if $n > max\{2m+d, 2d+m, b\} + 2b$.
\end{itemize}
\end{theorem}

\begin{proof}

First consider $n > \max \{ 2 m + d, 2 d + m, b \} + 2 b$. From Lemma \ref
{lemma-lm3}, for any scenario $\sigma$ we have $L M_3 (\sigma^{p_0}_p) =
\sigma (p_0)$ provided that $p_0$ is non-Byzantine. By definition $\sigma_p^{p_0} =
\sigma_p$, so we have $L M_3 (\sigma_p) = \sigma (p_0)$. This means every
non-Byzantine process could get the initial value of the non-Byzantine transmitter $p_0$
despite that $p_0$ might be partial faulty. Thus by applying $L M_3$ to a
$3$-round scenario $\sigma$, we could obtain a $3$-round reliable broadcast
algorithm. If we use this reliable broadcast algorithm as a broadcast
primitive in the early deciding algorithms in
\cite{dolev1982polynomial,krings1999byzantine}, then we could get early
deciding Byzantine agreement for a (n, m, d, b)-system as well. The time
complexity of algorithms in \cite{dolev1982polynomial,krings1999byzantine} is
$min \{ f+2, b + 1 \}$. Since we replace one round broadcast with three rounds
broadcast, the time complexity of early deciding algorithm with $3$-round
reliable broadcast is $min \{ 3(f + 2), 3(b + 1) \}$.

The result for $n \geqslant max\{2m+2d, b + 1\} + 2b$ follows from the same idea.
\end{proof}

\section{The Eventually Synchronous Case}\label{sec-esyn}

We considered so far synchronous computations. However, it is
also possible to tolerate partial failures in eventually synchronous systems. In this
section, we first present a reliable broadcast implementation that tolerates
partial Byzantine failures. Here, reliable broadcast ensures that if a non-Byzantine
process broadcasts a message then other processes will receive the
same message eventually (no such guarantee for Byzantine processes).
This broadcast primitive thus can be plugged into an
algorithm like \cite{castro1999practical}.

We assume here that after an unknown but finite time the
system become synchronous \cite{lynch1996distributed}.
Within an eventually synchronous system, the processes could not distinguish
message delay from the absence of a message. We consider a \tmem{static} (n,
m, d, b)-system which includes up to $b$ Byzantine processes and up to $m$
partial faulty processes each of which is associated with up to $d$
\tmem{fixed} Byzantine links. We first show that the algorithm $L M_2$ and $L
M_3$ can be modified to achieve reliable broadcast in an eventually synchronous (n, m,
d, b)-system.

\begin{algorithm}[h]
  {\tmstrong{Assume}}: $\sigma_p$ is a $2$-round view of process $p$ for a static
  (n, m, d, b)-system with $k \geqslant 2$ and $p_0$ is the
  transmitter.{\hspace*{\fill}}

  {\tmstrong{Code}} for $p$:
  \begin{fkindent}
  \begin{enumeratenumeric}
    \item Waits until receiving more than $n - m - d - b$ values for
    $\{ \sigma_p (p_0 p_1) : {p_1 \in P \backslash p_0} \}$ with a same value $v$, then
    output $v$.
  \end{enumeratenumeric}
  \end{fkindent}

  \caption{$2$-round Reliable-Broadcast ($R B_2$)\label{algo-rb2}}
\end{algorithm}

\begin{lemma}
  In $R B_2$ (Algorithm \ref{algo-rb2}), if $n \geqslant 2 m + 2 d + 2 b$ and $p_0$
  is non-Byzantine, then $R B_2 (\sigma_p) = \sigma (p_0)$.
\end{lemma}

\begin{proof}
  As in Lemma \ref{lemma-lm2}, $p_0$ will receive $n - m - d - b$ $\sigma_p
  (p_0 p_1)$ that equal to $\sigma (p_0)$ from $n - m - d - b$ correct
  processes. Since $2 ( n - m - d - b ) >
  n - 1$, the lemma follows.
\end{proof}

\begin{algorithm}[h]
  {\tmstrong{Assume}}: $\sigma_p$ is a $3$-round view of process $p$ for a static
  (n, m, d, b)-system with $k \geqslant 3$ and $p_0$ is the
  transmitter.{\hspace*{\fill}}

  {\tmstrong{Code}} for $p$:
  \begin{fkindent}
  \begin{enumeratenumeric}
    \item $p$ initializes an empty set $S$.
    
    \item Waits until receiving $n - m - b - 1$ values for $\{ \sigma_p 
    (p_0 p_1 p_2) : {p_2 \in P \backslash p_1} \}$ with a same value $v$, then $p$ adds 
    $v$ to $S$.
    
    \item Waits until $n - m - d - b$ values in $S$ have a same value $v'$,
    then $p$ outputs $v'$.
  \end{enumeratenumeric}
  \end{fkindent}
  \caption{$3$-round Reliable-Broadcast ($R B_3$)\label{algo-rb3}}
\end{algorithm}

\begin{lemma}
  In $R B_3$ (Algorithm \ref{algo-rb3}), if $n > \max \{ 2 m + d, 2 d + m, b \} + 2 b$
  and $p_0$ is non-Byzantine, then $R B_3 (\sigma_p) = \sigma (p_0)$.
\end{lemma}

\begin{proof}
  If $p_1$ is correct, there are at least $n - m - b
  - 1$ values equal to $\sigma (p_0 p_1)$ in $\{ \sigma_p (p_0 p_1 p_2)
  : {p_2 \in P \backslash p_1} \}$ from correct processes, which implies
  $\sigma (p_0 p_1)$ will be added to $S$ eventually.
  
  If $p_1$ is d-faulty, there are at most $m + d + b -
  1$ values different from $\sigma (p_0 p_1)$ in $\{ \sigma_p (p_0 p_1 p_2)
  : {p_2 \in P \backslash p_1} \}$. Since $n - m - b - 1 \geqslant m + d + b - 1$, only
  $\sigma (p_0 p_1)$ might be added to $S$.
  
  Now consider the transmitter.
  If $p_0$ is non-Byzantine, all correct processes
  except the one receiving wrong values from $p_0$ will contribute a value
  $\sigma (p_0)$ to $S$. So $S$ will eventually include at least $n - m - d - b$
  values equal to $\sigma (p_0)$ and at most $d + b$ values different from
  $\sigma (p_0)$. Since $n > m + 2 d + 2 b$, $R B_3 (\sigma_p)$ can only be
  $\sigma (p_0)$.
\end{proof}

If $R B_2$ or $R B_3$ is employed as a broadcast primitive, i.e. a
process broadcasts a message by executing an instance of $R B_2$ or
$R B_3$, then the messages broadcast by non-Byzantine processes will be
received by other non-Byzantine processes as if there are no partial failures. In
this way, $R B_2$ and $R B_3$ could play the role of reliable
broadcast for a (n, m, d, b)-system. We could then use our reliable broadcast
primitive (either RB2 or RB3) within an algorithm such as PBFT \cite{castro1999practical}.
We obtain the following theorem.

\begin{theorem}
  Byzantine agreement can be solved assuming eventually synchronous computation of a static (n, m, d,
  b)-system ($m, d > 0$) if and only if $n > \max \{2 m + d, 2 d + m, b\}+ 2
  b$.
\end{theorem}

\begin{proof}   
  The sufficiency follows from the above discussion. The
  necessity comes from Lemma \ref{lemma-impsb1} and Lemma \ref{lemma-impsb2}.
\end{proof}

\end{document}